\documentclass[10pt,twocolumn]{IEEEtran}
\usepackage[all]{xy}
\usepackage{color}
\usepackage{algorithm}
\usepackage{algpseudocode}
\usepackage{subfigure}
\usepackage{bm}
\usepackage{graphicx} 
\usepackage{color}
\usepackage{wrapfig}
\usepackage{epsf}
\usepackage{times}
\usepackage{url}
\usepackage{amssymb, amsmath, mathtools, amsthm}
\usepackage{nomencl}
\usepackage{comment}
\makenomenclature
\newtheorem{theorem}{Theorem}

\newtheorem{lemma}{Lemma}
\newtheorem{assumption}{Assumption}
\definecolor{RED}{rgb}{0.6,0.,0.}
\definecolor{BLUE}{rgb}{0.,0.,0.6}
\definecolor{GREEN}{rgb}{0.,0.6,0.}
\definecolor{MALINA}{rgb}{0.6,0.,0.6}
\definecolor{YELLOW}{rgb}{0.8,0.8,0}

\newcommand{\createspace}{\vspace{1.5 mm}}

\newcommand{\snr}{\textsc{SNR}}
\newcommand{\s}{\textsc{s}}

\begin{document}
\bstctlcite{IEEEexample:BSTcontrol}
\title{Tractable learning in under-excited power grids}
\author{
\IEEEauthorblockN{Deepjyoti Deka$^{\dag}$, Harish Doddi$^{\ddag}$, Sidhant Misra$^{\dag}$, Murti Salapaka$^{\ddag}$\\
\IEEEauthorblockA{(${\dag}$) Theoretical Division, Los Alamos National Laboratory, New Mexico, USA \\
(${\ddag}$) Department of Electrical and Computer Engineering, University of Minnesota, Minneapolis, USA}
}}

\maketitle
\begin{abstract}
Estimating the structure of physical flow networks such as power grids is critical to secure delivery of energy. This paper discusses statistical structure estimation in power grids in the `under-excited' regime, where a subset of internal nodes do not have external injection. Prior estimation algorithms based on nodal potentials/voltages fail in the under-excited regime. We propose a novel topology learning algorithm for learning under-excited general (non-radial) networks based on physics-informed conservation laws. We prove the asymptotic correctness of our algorithm for grids with non-adjacent under-excited internal nodes. More importantly, we theoretically analyze our algorithm's efficacy under noisy measurements, and determine bounds on maximum noise under which asymptotically correct recovery is guaranteed. Our approach is validated through simulations with non-linear voltage samples generated on test grids with real injection data. 
\end{abstract}
\begin{IEEEkeywords}
Flow conversation, Inverse graph Laplacian, noisy regression, Matpower, inverse covariance graph, distribution grid. 
\end{IEEEkeywords}

\section{Introduction}
\label{sec:intro}
Topology estimation is a crucial part of situational awareness in power grids, which is the backbone for energy needs in the modern economy \cite{dekatcns}. Operationally the power grid can be distinguished between high voltage transmission grids and low voltage distribution grids \cite{kersting2001radial}. Structurally, the transmissions are meshed/loopy (have cycles), while distributional grids are primarily radial (tree-like). In either setting, the grid topology is established by the status of breakers/switches on an underlying set of permissible lines/edges as shown in Fig.~\ref{fig:city}. In the modern grid, presence of stochastic resources such as solar energy, wind farms, and aggregated controllable loads have made topology and state estimation with guarantees of paramount importance for secure operation. In the absence of sufficient line measurements \cite{hoffman2006practical}, in recent years, researchers \cite{arghandeh2017big} have proposed the use of nodal voltage measurements collected from new, high-fidelity nodal measurement devices such as PMUs \cite{PMU}, micro-pmus \cite{micropmu}, and advanced metering equipment to aid in topology estimation. Among algorithms with theoretical guarantees, passive greedy schemes \cite{dekatcns,sandia1,sejunpscc} and active probing methods \cite{cavraro2018graph,cavraro2019inverter} that employ voltage measurements are proposed for learning radial grids. However, they do not generalize to meshed/loopy networks in transmission grids and urban distribution grids \cite{NY,germany}. Similarly, algorithms for joint estimation of topology and line parameters using both voltage and injection measurements have been designed using least-squared fitting \cite{yu2018patopaem} and low-ranked decomposition \cite{yuan2016inverse}. While these works are related, in this article we focus on the problem of topology recovery using only nodal voltage measurements without any additional measurement or statistics of nodal injections. In a related line of work, consistent algorithms based on inverse covariance or mutual independence (MI) of nodal voltages are used for learning grid topology in both radial \cite{he2011dependency,bolognani2013identification,dekathreephase}, and loopy settings \cite{liao2018urban,dekaloopy}. Its theoretical generalizability notwithstanding, the validity of inverse voltage covariance or MI based methods rely on a key assumption: all nodes/buses are excited, i.e., have non-zero injection fluctuations. This is necessary to ensure that the voltage covariance matrix is full-ranked and hence invertible. While terminal nodes in grids generally have loads or generation, an unidentified set of internal nodes with zero injection may be present in realistic transmission and distribution grids for routing power flows \cite{kersting2001radial}. In such setting, voltage measurement based learning algorithms proposed in prior work do not extend theoretically. The overarching goal of this work is to overcome this drawback and develop a provably correct topology learning algorithm using only \emph{nodal voltages} from a \emph{general (loopy)} power grid, where an \emph{unidentified} subset of internal nodes may be \emph{unexcited}. As linearized power flow models are examples of structural equation models (SEM) \cite{bollen2005structural}, our work enables tractable learning for SEMs satisfying conservation laws, in the under-excited regime, using only state samples. Extensions in systems beyond power grids are the objective of our future work.

\subsection{Contribution}
Our learning algorithm is based on the physics of linearized power flows. As the voltage covariance matrix is rank-deficient in the under-excited regime, we propose a two-step learning approach. In the \emph{first} step, we develop a regression based test that identifies the internal nodes of zero injection, and subsequently learns their true neighbors. In the \emph{second} step, we extract the voltage covariance pertaining only to nodes with non-zero injection, and estimate the rest of the network. We prove that under noiseless measurements, our algorithm correctly estimates the underlying topology of radial grids, provided under-excited internal nodes are non-adjacent. For loopy grids of minimum loop-size four, our algorithm ensures correct learning provided loops of size four do not include zero-injection buses. Note that the minimum loop size of four is non-restrictive as real grids have large girth \cite{NY}, and is necessary for correct estimation even in the fully-excited setting \cite{dekaloopy}. We show the necessity of our assumptions for unique and consistent recovery, through multiple counter examples. Under noisy data, we determine the maximum noise margin for consistency of each step of our learning algorithm, and express this margin as an intuitive function of system parameters. We validate the developed algorithms through experiments on both linear and non-linear voltage samples collected from test grids with real injection data. 

The next section discusses power flows in grids and related voltage statistics. Section~\ref{sec:zero} discusses properties of nodal voltages which is used to extract zero-injection buses and their neighbors. Section~\ref{sec:non-zero} determines edges between non-zero injection buses and completes the topology estimation. The analysis under noise for each step is provided in its corresponding section. Section~\ref{sec:simulations} includes simulations results on IEEE test networks, using non-linear samples generated using real injection data. Conclusions and future work, are included in Section~\ref{sec:conclusion}.
\begin{figure}[thb]
\centering
\includegraphics[width=0.45\textwidth]{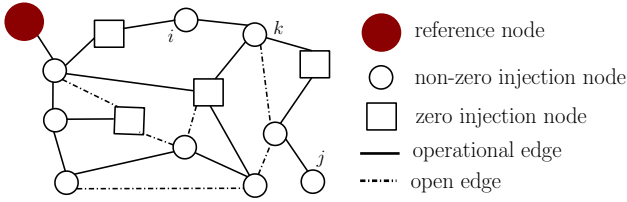}
\caption{Loopy power grid with internal ($i,k$) and terminal ($j$) nodes.
\label{fig:city}}
\vspace{-3mm}
\end{figure}

\section{Preliminaries}
\label{sec:structure}
\subsection{Notation}
We denote the power grid topology by an undirected graph $\mathcal{G}\equiv(\mathcal{V},\mathcal{E})$, where $\mathcal{V} =\{i,j,..\}$ is the set of $N+1$ nodes/nodes, and $\mathcal{E} = \{(ij), (jk),...\}$ is the set of operational lines/edges (see Fig.~\ref{fig:city}). Two nodes that share an edge are termed `neighbors'.
The degree of a node is the total number of edges that involve the node. A node with degree $1$ is termed a `terminal node', and an `internal' node otherwise. A `loop' of size $l \geq 3$ refers to a set of distinct edges $\{(ij_1),(j_1j_2),...,(j_{l-1}i)\}$. A radial (loop-less) graph is defined to have loops of minimum size $\infty$. Each edge $(ij)$ in grid $\mathcal{G}$ is associated with susceptance $\beta_{ij} >0$, and conductance $g_{ij}>0$. Each node $i$ has a corresponding voltage magnitude $v_i$, phase angle $\theta_i$, active power $p_i$ and reactive power $q_i$. We use $v,\theta,p,q$ to denote the corresponding vectors with entries for each node. 

\subsection{Power grid basics}
The Kirchoff's law for AC power flow (AC-PF) is given by:
\begin{align}
\forall i \in \mathcal{V}, ~~p_i+\hat{i} q_i= \smashoperator[lr]{\sum\limits_{j:(ij)\in\mathcal{E}}}(v_i^2-v_i v_je^{\hat{i}\theta_i-\hat{i}\theta_j})(g_{ij}+\hat{i}\beta_{ij}). \label{P-complex1}
\end{align}
Considering small deviations in $v$ from a reference bus voltage and small differences in $\theta$ between neighboring nodes, the non-linear AC-PF can be linearized in $v,\theta$ to get the \emph{Linearized AC} power flow (LC-PF) \cite{dekatcns,89BWc,bolognani2016existence} whereby, 
\begin{align} 
p+\hat{i}q = H_{g+\hat{i}\beta}(v-\hat{i}\theta)\text{ with }\begin{bmatrix}p \\q\end{bmatrix} = \begin{bmatrix}H_g &H_{\beta}\\ H_{\beta} &-H_{g}\end{bmatrix}\begin{bmatrix}v \\\theta\end{bmatrix}. \label{LC_PF}
\end{align}
The susceptance weighted Laplacian matrix $H_{\beta}$ is given by
\begin{align}
 H_\beta(i,j) = \begin{cases}\sum_{k:(ik) \in {\cal E}}\beta_{ik} ~~\text{if~} i=j\\
-\beta_{ij} ~~\text{if~} (ij) \in {\cal E}\\
 0 ~~\text{otherwise}\end{cases}.\label{Laplacian}
\end{align}
$H_g$ is defined similarly. If voltage magnitude deviations are negligible, or lines are primarily susceptive, the linear \emph{DC power flow} (DC-PF), is derived \cite{89BWc}: 
\begin{align}
 p_i = \smashoperator[lr]{\sum_{j:(ij)\in \mathcal{E}}}\beta_{ij}(\theta_i-\theta_j)\Rightarrow p = H_\beta\theta. \label{DC_PF}
\end{align}
Note that both LC-PF and DC-PF models are lossless. One node's voltage can be taken as reference to measure all other voltages, while its injection is given by the negative sum of all other injections. The PF models can then be reduced to the $N$ non-reference nodes by removing entries in $p,q,v,\theta$, and rows and columns in $H_{\beta}, H_{g}$ that correspond to the reference node. From this point onwards, we use $p,q,v,\theta,H_{\beta}, H_g$ to refer to their reduced versions. As reduced $H_{\beta},H_{g}$ have full rank, LC-PF and DC-PF models are invertible. In the rest of the article, we focus on DC-PF for our theoretical results. Our analysis naturally extends to the LC-PF model, as noted in subsequent sections. Furthermore, our simulations results demonstrate the performance of our learning algorithm on non-linear AC-PF voltage measurements.

\subsection{Structure recovery problem}
The structure recovery problem pertains to estimating the set of edges $\mathcal{E}$ of $\mathcal{G}$ using the measurements $\{\theta_{i}^t, \ i \in \mathcal{V}, \ t=1,\ldots,T\}$ of the phase angles (and voltages for the LC-PF model). For the theoretical analysis, we consider the asymptotic setting $T \rightarrow \infty$, that ensures sufficient number of collected voltage observations for correct estimation of the second order statistics of the nodal voltages. Our simulation results demonstrate the performance in the low-sample regime.
In an under-excited grid $\mathcal{G}$, we denote 
the set of internal nodes with zero-injection by $\mathcal{U}$, and the remaining nodes in $\mathcal{V}-\mathcal{U}$ by ${\mathcal{U}_c}$, respectively. Without a loss of generality, we consider the nodes numbered $1$ to $|\mathcal{U}|$ ($|\mathcal{U}| $ denotes the cardinality of set $\mathcal{U}$) as zero-injection and rest as non-zero injection nodes. Thus, we partition $p = \begin{bmatrix}\mathbf{0}\\p^{\mathcal{U}_c}\end{bmatrix}$, and $\theta = \begin{bmatrix}\theta^\mathcal{U}\\\theta^{{\mathcal{U}_c}}\end{bmatrix}$. We assume the following regarding covariance nodal injections:

\begin{assumption} \label{a:uncorrelated_p}
Fluctuations in $p^{\mathcal{U}_c}$ are uncorrelated:
\begin{align*}
  \Sigma_{p} = \begin{bmatrix}
 0~&0\\0&\Sigma_{p^{\mathcal{U}_c}}
 \end{bmatrix}
\end{align*}
and $\Sigma_{p^{\mathcal{U}_c}}$ $=\mathbb{E}[p^{\mathcal{U}_c}{p^{\mathcal{U}_c}}^T]\succ0$ is diagonal.
\end{assumption}
Assumption~\ref{a:uncorrelated_p} is used in prior work \cite{he2011dependency,bolognani2013identification,dekathreephase} for tractable grid learning. Further it also holds for real loads and renewables \cite{liao2018urban,yury}, when small linear trends are empirically de-trended. While Assumption~\ref{a:uncorrelated_p} is used for proving the correctness of a part of our algorithm (Section~\ref{sec:non-zero}), we consider real correlated injection data in our numerical simulations to demonstrate accurate learning, in its absence.

The relevant statistics is the covariance matrix of the observables $ \Sigma_{\theta}$, which under the DC-PF~\eqref{DC_PF} is given by 
\begin{align}
 \Sigma_{\theta} = \mathbb{E}[\theta\theta^T] = J_{\beta} \begin{bmatrix}
 0~&0\\0&\Sigma_{p^{\mathcal{U}_c}}
 \end{bmatrix}J_{\beta}, \text{~where~} J_\beta = H^{-1}_\beta. \label{covar_DC}\end{align}
A similar formula for the voltage covariance $\Sigma_{(v,\theta)} = \begin{bmatrix}\Sigma_{v}~&\Sigma_{v\theta}\\\Sigma_{\theta v}~&\Sigma_{\theta}\end{bmatrix}$ can be derived for the LC-PF~\eqref{LC_PF}. 
\subsection{The fully excited case}
If $\mathcal{U}$ is empty, $\Sigma_p$ and $\Sigma_\theta$ are invertible, and 
\begin{align}
&\Sigma^{-1}_{\theta} = H_{\beta}\Sigma^{-1}_{p}H_{\beta}, \text{~where using~\eqref{Laplacian}, we verify that}\nonumber\\
&\Sigma^{-1}_{\theta}(i,j)\coloneqq \begin{cases}&\hspace{-9pt}<0~~\text{if $(ij) \in \mathcal{E}$, and $\forall k, (ik),(jk) \notin \mathcal{E}$}\label{DC_full_inv}\\
&\hspace{-9pt}>0~~\text{if $(ij) \notin \mathcal{E}$, and $\exists k$, $(ik),(jk) \in \mathcal{E}$}\\
 &\hspace{-9pt}=0 ~~\text{if $(ij) \notin \mathcal{E}$, and $\forall k, (ik),(jk) \notin \mathcal{E}$.}\end{cases}
\end{align}
Note that if $i,j$ are not part of a loop of size three, then they must exist in one of the configurations in Eq.~\eqref{DC_full_inv}. The following thus holds.
\begin{lemma}\cite{dekaloopy}\label{lemma:full_inv}
Let nodes $i,j$ in fully-excited $\mathcal{G}$ not be part of a three-node loop. Under Assumption~\ref{a:uncorrelated_p}, edge $(ij)$ exists if and only if $\Sigma^{-1}_{\theta}(ij) <0$.
\end{lemma}
\textbf{Note:} 
In the LC-PF model \cite{dekaloopy}, under faithfullness assumptions, it can be shown that if $i,j$ are not in a three-node loop, $(ij)$ exists \textit{iff} $\Sigma^{-1}_{(v,\theta)}(i,j)+\Sigma^{-1}_{(v,\theta)}(i+N,j+N) <0$.
 
If the minimum loop size for grid $\mathcal{G}$ is four, no node pair exists in a three-node loop. Thus, Lemma~\ref{lemma:full_inv} enables \emph{correct topology estimation in the fully-excited setting} \cite{bolognani2013identification,dekaloopy}, under both DC-PF and LC-PF. Noteably, Lemma~\ref{lemma:full_inv} fails in the under-exited regime as $\Sigma_{\theta}$ or $\Sigma_{(v,\theta)}$ are not invertible. The rest of the paper discusses tractable learning in this under-excited setting. To learn grids with zero-injection node set $\mathcal{U}$, we propose a two-step topology learning algorithm. First, we identify zero injection nodes and their neighbors. Next, we estimate edges between node pairs with non-zero injection. The next section presents a regression-based framework for the first step.

\section{Identifying zero-injection nodes and neighbors}
\label{sec:zero}
We first start with stating the required assumption for structure recovery in the under-excited setting.

\subsection{Structural Assumptions} 
Terminal nodes in power grids have loads or generators and generally have non-zero injection. For theoretical completion, consider a zero-injection terminal node $i$ with unique neighbor $j$ in $\mathcal{G}$. From Eq.~\eqref{DC_PF}, $\theta_i =\theta_j$. Thus any voltage-based learning algorithm can estimate $\mathcal{G}$ only up to a permutation between $i,j$, as $i,j$ can be interchanged without affecting the algorithm's output. However the pair $i,j$ can be checked using $(\theta_i-\theta_j)^2=0$. We thus identify node groups with overlapping phase angles. To enable unique estimation, we retain one node per group as neighbor and remove the rest labelled as terminal zero-injection nodes, before estimating the rest of the network. This, along with another structural assumption required for our regression approach, is stated below.
\begin{assumption} \label{a:terminal_node}
In grid $\mathcal{G}$, all zero-injection nodes are internal (degree $>1$) and non-adjacent.
\end{assumption}
The non-adjacency assumption prevents cases where a large fraction/majority of internal nodes have no injection. We subsequently show through counter examples, that in the absence of Assumption \ref{a:terminal_node}, both zero-injection buses and their neighbors may be incorrectly identified. We first consider the setting with noiseless voltage measurements.

\subsection{Noiseless setting}
To begin our analysis, we partition $H_\beta$ ($J_{\beta}$) into rows (columns) corresponding to $\mathcal{U}$ and $\mathcal{U}_c$:
$H_\beta= \begin{bmatrix}H^{\mathcal{U}}_\beta\\H^{\mathcal{U}_c}_\beta\end{bmatrix}$, $J_{\beta}= [J^{\mathcal{U}}_{\beta}~|~J^{\mathcal{U}_c}_{\beta}]$. As $H_\beta J_\beta = \mathbb{I}$, $ H^{\mathcal{U}}_\beta J^{\mathcal{U}_c}_\beta = \mathbf{0}$. Thus,
\begin{align}
y^TJ^{\mathcal{U}_c}_\beta = 0 \iff y^T = c^TH^{\mathcal{U}}_\beta \text{~for some $c$}.\label{nullspace}
\end{align}
\eqref{nullspace} follows from the fact that both rows-space of $H^{\mathcal{U}}_\beta$ and left null-space of $J^{\mathcal{U}_c}_\beta$ have rank $|\mathcal{U}|$. Following Eq.~\eqref{Laplacian} and Assumption \ref{a:terminal_node}, the matrix $H^{\mathcal{U}}_\beta$ satisfies for the following properties:
\begin{subequations}\label{H_dd}
\begin{align}
  &H^{\mathcal{U}}_\beta(i,i) > 0, \quad \forall i \in \mathcal{U} \\
  &H^{\mathcal{U}}_\beta(i,j) = 0, \quad \forall i \in \mathcal{U}, \ \forall j \in \mathcal{U}-\{i\} \\
  &H_\beta(i,k)\leq 0, \quad \forall i \in \mathcal{U}, \ \forall k \in \mathcal{U}_c \\
  &\forall i \in \mathcal{U}, \ \exists \{k_1^i, k_2^i\} \in \mathcal{U}_c \ \mbox{such that} \nonumber \\
  &\qquad H_\beta(i,k) = -\beta_{ik}<0, \ k \in \{k_1^i, k_2^i\}.
\end{align}
\end{subequations}


Let $\theta_{-i}$ represent the vector of phase angles at all nodes but $i$. Consider the following regression problem, solved for each $i \in \mathcal{V}$:

\createspace\noindent \textbf{Nodal Regression:}
\begin{align} \label{P:regression_nonoise_1}
  a^{i*} = \arg\smashoperator[lr]{\min_{x \in\mathbb{R}^{N-1}}} \mathbb{E}[(\theta_i -\theta^{T}_{-i}x)^2]\text{~s.t.~} x \geq \mathbf{0},~\mathbf{1}^Tx \leq 1.
\end{align}
The next result shows that~\eqref{P:regression_nonoise_1} can identify nodes in $\mathcal{U}$.


\begin{theorem}\label{thm:U_identify_1}
In grid $\mathcal{G}$, under Assumption~\ref{a:terminal_node}, $i$ is a zero-injection bus in $\mathcal{U}$ if and only if the minimum value for Problem~\eqref{P:regression_nonoise_1} is zero.
\end{theorem}
Once the nodes in $\mathcal{U}$ have been identified, we use the following regression to identify the \emph{neighbors} of each node $i\in \mathcal{U}$.

\createspace\noindent \textbf{Constrained Nodal Regression:}
\begin{align}
 b^{i*} = \arg\smashoperator[lr]{\min_{x \in\mathbb{R}^{\mathcal{U}_c}}} \mathbb{E}[(\theta_i -{\theta^{\mathcal{U}_c}}^Tx)^2]\text{~s.t.~} x\geq \mathbf{0},~\mathbf{1}^Tx \leq 1
\label{P:regression_nonoise_2}
\end{align}
\begin{theorem}\label{thm:U_identify_2}
For $i \in \mathcal{U}$, consider Problem~\eqref{P:regression_nonoise_2}. 
Under Assumption~\ref{a:terminal_node}, $b^{i*}$ has cost $0$, and neighbors of $i$ are given by $\{j:j\in \mathcal{U}_c, b^{i*}_{j-|\mathcal{U}|} \neq 0\}$. 
\end{theorem}
The regression problem~\eqref{P:regression_nonoise_2} differs from~\eqref{P:regression_nonoise_1} in that the regression vector $x$ is now constrained to be zero for nodes in $\mathcal{U}$, that have been identified using Theorem~\ref{thm:U_identify_1}. This constraint leverages the structural property in Assumption~\ref{a:terminal_node}. We now prove both these results.
\createspace
\begin{proof}[Proof of Theorem~\ref{thm:U_identify_1}]
Consider $y\in\mathbb{R}^N$ with $y_i =1$, $y_{-i} = -x$. The constraints and cost in Problem~\eqref{P:regression_nonoise_1} are re-formulated as:
\begin{flalign}
 &\text{(Constraints:)~~~~}y_i = 1,~y_{-i}\leq 0,~\mathbf{1}^Ty\geq 0 \label{pf1}\\
&\text{(Cost:)~~~~~~}\mathbb{E}(\theta y)^2 =y^T\Sigma_{\theta}y = y^TJ^{\mathcal{U}_c}_{\beta}\Sigma_{p^{\mathcal{U}_c}}{J_{\beta}^{\mathcal{U}_c}}^Ty\nonumber\\
&\text{Note that, (using Eq.~\eqref{nullspace}), }\nonumber\\
&~~~~~~~~~~~\mathbb{E}(\theta y)^2 = 0 \iff y^T = c^TH^{\mathcal{U}}_\beta \text{ for some $c.$}&&\label{pf2}
\end{flalign}
For $i\in \mathcal{U}$, let $y^T= \frac{H^{\mathcal{U}}_\beta(i,:)}{H_\beta(i,i)}$. From Eq.~\eqref{pf2}, $\mathbb{E}(\theta y)^2 =0$. Using Eq.~\eqref{Laplacian}, $y_i = 1$ , while $y_j \leq 0~\forall j\neq i$ and $\mathbf{1}^Ty =0$. Thus a feasible $y$ with cost $0$ exists for $i\in \mathcal{U}$.

For $i\in \mathcal{U}_c$, suppose $y^T = c^TH^{\mathcal{U}}_\beta$ with $y_i = 1, y_{-i}\leq0$. Note that for all $j\in \mathcal{U}$, $j\not =i$ and using~\eqref{H_dd}, $y_j = c_jH_{\beta}(j,j)$. $y_j\leq 0$ thus enforces $c_j\leq 0$ for all $j\in \mathcal{U}$. Similarly, $y_i =1$ implies $\exists j_1\in\mathcal{U}$ such that $c_{j_1}\neq 0$ and is strictly negative. Under Assumption~\ref{a:terminal_node}, $\exists k \in \mathcal{U}_C-\{i\}$ such that edge $(j_1k) \in \mathcal{E}$. Then, $y_k = -c_{j_1}\beta_{j_1k} + \sum_{j\in \mathcal{U}-\{j_1\}}c_jH_\beta(j,k)$. As $H_\beta(j,k)\leq 0$, we have $y_k>0$, which violates $y_{-i}\leq0$. Thus no $y = c^TH^{\mathcal{U}}_\beta$ is feasible with Eq.~\eqref{pf1} and the cost in Problem~\eqref{P:regression_nonoise_1} is non-zero for $i \in \mathcal{U}_c$. 
\end{proof}
\createspace
\begin{proof}[Proof of Theorem~\ref{thm:U_identify_2}]
For $i\in\mathcal{U}$, all its neighbors belong to $\mathcal{U}_c$, under Assumption~\ref{a:terminal_node}. Consider the following reformulation of the cost and constraints of Problem~\eqref{P:regression_nonoise_2}, 
\begin{align}
&\text{(Cost:)~~}\mathbb{E}(\theta y)^2 = y^TJ^{\mathcal{U}_c}_{\beta}\Sigma_p^{\mathcal{U}_c}{J_{\beta}^{\mathcal{U}_c}}^Ty \text{~with constraints,}\\
&y_i = 1,~y_j =0 \forall j \in \mathcal{U}-\{i\},~ y_{-i}\leq 0,~\mathbf{1}^Ty\geq 0. \label{pf2_2}
\end{align}
Suppose $y$ is a solution with cost zero. Using~\eqref{pf2}, there exists $c$ such that , $y^T=c^TH^{\mathcal{U}}_\beta$. As nodes in $\mathcal{U}$ are non-adjacent, $y_j = c_jH_\beta(j,j)$ for $j\in\mathcal{U}$. To ensure $y_i = 1$ and $y_j =0$, for all $j\in\mathcal{U}-\{i\}$, we have $c_i = 1/H_\beta(j,j)$ and $c_j= 0$ for all $ j \neq i$. Thus, $y^T= \frac{H^{\mathcal{U}}_\beta(i,:)}{H_\beta(i,i)}$ is the unique solution. From~\eqref{H_dd}, it follows that $y$ satisfies constraints in~\eqref{pf2_2}. Note that $j \neq i$, $y_j = H_\beta(i,j)/H_\beta(i,i)$ and thus is non-zero if and only if edge $(ij) \in\mathcal{E}$. Thus, it identifies the neighbors of $i$. 
\end{proof}

\textbf{Note:} (a) The constraints in Problem~\eqref{P:regression_nonoise_1} stem from power-flow conservation law. Using a similar analysis for $(v,\theta)$ in the LC-PF Eq.~\eqref{LC_PF}, it can be shown that nodes in $\mathcal{U}$ are identified by zero-cost solutions of the following complex-valued regression problem. 
\begin{align}
 \arg\smashoperator[lr]{\min_{x \in\mathbb{C}^{N-1}}} &\mathbb{E}[|(v_i-\hat{i}\theta_i -(v_{-i}-\hat{i}\theta_{-i})^Tx)|^2]\label{complex_reg}\\
 \text{~s.t.~}&\text{Real}(x) \geq \mathbf{0},~\mathbf{1}^T\text{Real}(x) \leq 1,\nonumber\\
 -\mathbf{1}&\leq \text{Imag}(x) \leq \mathbf{1},~\mathbf{1}^T\text{Imag}(x) =0.\nonumber
 \end{align} 
 (b) Following Eq.~\eqref{complex_reg}, the LC-PF version of Theorem \ref{thm:U_identify_2} can be constructed as a complex-valued regression problem over $(v,\theta)$, with regression coefficients restricted to $\mathcal{U}_c$. 

Note that Assumption~\ref{a:uncorrelated_p} is not used in either identification of nodes in $\mathcal{U}$ (Theorem~\ref{thm:U_identify_1}), or estimation of their neighbors (Theorem~\ref{thm:U_identify_2}). We next show that non-adjacency of nodes in $\mathcal{U}$ (Assumption~\ref{a:terminal_node}) is necessary for both these steps.
\begin{figure}[htb]
\centering\hfill
\includegraphics[width=0.45\textwidth]{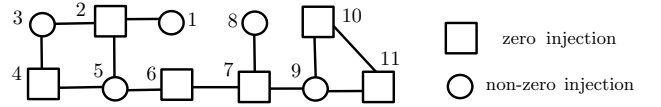}\hfill
\caption{Grid that violates Assumption~\ref{a:terminal_node} due to nodes $6,7$ and $10,11$.\label{fig:wrong_ab}}
\end{figure} 

\textbf{Requirement of Assumption~\ref{a:terminal_node}:} Consider the grid graph in Fig.~\ref{fig:wrong_ab}, with adjacent node pairs $6,~7$ and $10,~11$. 

Consider Problem~\eqref{P:regression_nonoise_1} for node $9 \in \mathcal{U}_c$. By combining rows $H^T_{\beta}(10,:), H^T_{\beta}(11,:)$ of adjacent zero-injection nodes, we can create a feasible zero-cost solution $y^T$ for $i=9$, which will incorrectly identify $9$ as a zero-injection node. For example, if $\min(\beta_{9,10},\beta_{9,11})> \beta_{9,11}$, then $y^T = -\frac{H^T_{\beta}(10,:)}{2\beta_{9,10}} -\frac{H^T_{\beta}(11,:)}{2\beta_{9,11}}$ can be verified as a valid zero-cost solution for $i = 9$. 

Now consider node $6 \in \mathcal{U}$. It can be verified that zero-cost solutions in Problem ~\eqref{P:regression_nonoise_1} for $i =6$ are given by $$y^T = \frac{H_\beta(7,7)H^T_\beta(6,:)+\epsilon \beta_{67}H^T_{\beta}(7,:)}{H_\beta(6,6)H_\beta(7,7)-\epsilon\beta^2_{67}} \text{~with~}0\leq\epsilon\leq 1.$$ Taking $\epsilon=1$ ensures $y^T_7 =0$, and produces a zero-cost solution for Problem~\eqref{P:regression_nonoise_2}, that incorrectly estimates $5,8,9$ as neighbors of $6$. Assumption~\ref{a:terminal_node}
is thus necessary for correctness of Theorems \ref{thm:U_identify_1} and \ref{thm:U_identify_2}. Next, we analyze the asymptotic correctness of Theorems \ref{thm:U_identify_1} and \ref{thm:U_identify_2} in the practical setting, where measurements are noisy.

\subsection{Noisy setting} \label{sec:noise_1}
Let $\tilde{\theta}$ be the phase measurement corrupted by independent noise $n$ of mean zero and covariance $\Sigma_n \succ 0$, i.e.,
\begin{align}
 \tilde{\theta}\hspace{-2pt}=\hspace{-2pt}\theta+n \Rightarrow\hspace{-2pt} \Sigma_{\tilde{\theta}}= \Sigma_{\theta} +\Sigma_n. \label{eq_noise}
\end{align}
First, we define the following parameters, {
\begin{align}
\beta_{min}\hspace{-2pt} =\hspace{-2pt}\smashoperator[lr]{\min_{(ij) \in \mathcal{E}}} \beta_{ij},\s_{id}\hspace{-2pt}= \hspace{-2pt}\smashoperator[lr]{\max\limits_{i\in\mathcal{V}}}\frac{H_\beta(i,i)}{\beta_{min}}, \snr\hspace{-2pt}=\hspace{-2pt}\frac{\sigma_{min}(\Sigma_{\theta^{\mathcal{U}_c}})}{\sigma_{max}(\Sigma_{n})}. \label{beta_min}
\end{align}}
Here $\sigma_{min}(X)$ and $\sigma_{max}(X)$ denote the minimum and maximum singular values of matrix $X$ respectively. $\Sigma_{\theta^{\mathcal{U}_c}}$ is the covariance of $\theta^{\mathcal{U}_c}$.

Let the minimum for Problem~\eqref{P:regression_nonoise_1} for node $i$ under noise be achieved at $a_n^{i*}$. Let $y_n \in\mathbb{R}^N$, with ${y_n}_i = 1, {y_n}_{-i} = -a_n^{i*}$. Similarly, denote $y \in\mathbb{R}^N$, with $y_i = 1,y_{-i} = -a^{i*}$, where $a^{i*}$ is the noiseless optima. For $i\in\mathcal{U}$, optimality of $y_n$ under noise, gives,
\begin{align}
&{y_n}^T(\Sigma_{\theta} + \Sigma_n)y_n\leq y^T(\Sigma_{\theta} + \Sigma_n)y\stackrel{(a)}{=}y^T\Sigma_ny \nonumber\\
&\hspace{50pt}\leq (1+ \|a^{i*}\|_2^2)\sigma_{max}(\Sigma_{n})\stackrel{(b)}{\leq}2\sigma_{max}(\Sigma_{n})\label{ineq_U} 
\end{align}
Here $(a)$ follows from Theorem \ref{thm:U_identify_1} for $i \in\mathcal{U}$. $(b)$ follows from $\|a^{i*}\|_2\leq \|a^{i*}\|_1 \leq 1$. In contrast, when $i \in \mathcal{U}_c$, the following result holds:
\begin{lemma}\label{lem:boundU_c}
For $i \in \mathcal{U}_c$, consider $y_n\in\mathbb{R}^N$ where ${y_n}_i=1$, and ${y_n}_{-i} = -a_n^{i*}$ is the optimal solution for Problem~\eqref{P:regression_nonoise_1} under noise. Then, ${y_n}^T(\Sigma_{\theta} + \Sigma_n)y_n> \frac{\sigma_{min}(\Sigma_{\theta^{\mathcal{U}_c}})}{1+\s_{id}^4+\s^2_{id}}$, where $\s_{id}$ is given in~\eqref{beta_min}. 
\end{lemma}
The proof is provided in Appendix \ref{sec:lem:boundU_c}. The following result combines Lemma \ref{lem:boundU_c} with~\eqref{ineq_U} to ensure correct identification of $\mathcal{U}$ under noise.

\begin{theorem}\label{thm:U_identify_1_noise}
If $\snr\geq2(1+\s^4_{id}+\s^2_{id})$, only nodes in $\mathcal{U}$ have cost lower than $\frac{\sigma_{min}(\Sigma_{\theta^{\mathcal{U}_c}})}{1+\s_{id}^4+\s^2_{id}}$ in Problem~\eqref{P:regression_nonoise_1}, where $\s_{id}$, $\snr$ are defined in~\eqref{beta_min}.
\end{theorem}
 
We now analyze Theorem \ref{thm:U_identify_2} for noisy estimation of the neighbors of zero-injection nodes. For node $i\in\mathcal{U}$, let the optimal for Problem~\eqref{P:regression_nonoise_2} be achieved at $b_n^{i*}$ under noise, and at $b^{i*}$ in the noiseless setting. Let $\theta^\mathcal{M} = \begin{bmatrix}\theta^i\\\theta^{\mathcal{U}_c}\end{bmatrix}$, $n^\mathcal{M} = \begin{bmatrix}n^i\\n^{\mathcal{U}_c}\end{bmatrix}$, where $\mathcal{M} = \{i\}\cup\mathcal{U}_c$. By optimality of $b_n^{i*}$ in the noisy setting, we have,{\footnotesize
\begin{align}
&\begin{bmatrix} 1\\ -b_n^{i*} \end{bmatrix}^T\hspace{-5pt}(\Sigma_{\theta^\mathcal{M}} + \Sigma_{n^\mathcal{M}})\hspace{-2pt}\begin{bmatrix} 1\\-b_n^{i*} \end{bmatrix} \leq 
 \begin{bmatrix} 1\\-b^{i*} \end{bmatrix}^T\hspace{-5pt}(\Sigma_{\theta^\mathcal{M}} + \Sigma_{n^\mathcal{M}})\hspace{-2pt}\begin{bmatrix} 1\\-b^{i*} \end{bmatrix}\nonumber \\
&\Rightarrow\begin{bmatrix} 0\\ b_n^{i*}-b^{i*} \end{bmatrix}^T\hspace{-5pt}(\Sigma_{\theta^\mathcal{M}} + \Sigma_{n^\mathcal{M}})\hspace{-2pt}\begin{bmatrix} 0\\ b_n^{i*}-b^{i*} \end{bmatrix} \nonumber\\
 &\hspace{50pt}\leq 2\begin{bmatrix} 0\\ b_n^{i*}-b^{i*} \end{bmatrix}^T\hspace{-4pt}(\Sigma_{\theta^\mathcal{M}} + \Sigma_{n^\mathcal{M}})\begin{bmatrix}1 \\ -b^{i*} \end{bmatrix} \nonumber \\
 &\hspace{50pt}\stackrel{(a)}{=}2\begin{bmatrix} 0\\ b_n^{i*}-b^{i*} \end{bmatrix}^T\hspace{-4pt}\Sigma_{n^\mathcal{M}}\begin{bmatrix}1 \\ -b^{i*}\end{bmatrix} \stackrel{(b)}{\leq} 4\sqrt{2}\sigma_{max}(\Sigma_{n}) \label{eq:partial1}
\end{align}}
Here, $(a)$ follows from optimal cost of $0$ for $b^{i*}$ in the noiseless setting (Theorem \ref{thm:U_identify_2}). (b) follows from $\|b_n^{i*}-b^{i*}\|_2\leq \|b_n^{i*}-b^{i*}\|_1 \leq \|b_n^{i*}\|_1 + \|b^{i*}\|_1 \leq 2$, and $\sigma_{max}(\Sigma_{n^\mathcal{M}})\leq \sigma_{max}(\Sigma_{n})$. Further, as $\Sigma_{n} \succ 0$, we have
\begin{align}
&\begin{bmatrix} 0\\ b_n^{i*}-b^{i*} \end{bmatrix}^T\hspace{-7pt}(\Sigma_{\theta^{\mathcal{M}}} +\Sigma_{n^{\mathcal{M}}})\hspace{-3pt}\begin{bmatrix} 0\\ b_n^{i*}-b^{i*} \end{bmatrix} >\nonumber\\ &\hspace{120pt}(b_n^{i*}-b^{i*})^T\Sigma_{\theta^{\mathcal{U}_c}}(b_n^{i*}-b^{i*}) \nonumber\\
&\Rightarrow~4\sqrt{2}\sigma_{max}(\Sigma_{n})> \sigma_{min}(\Sigma_{\theta^{\mathcal{U}_c}}) \|b_n^{i*}-b^{i*}\|_{\infty}^2 .\label{delta_inf}
\end{align}
The last inequality follows from~\eqref{eq:partial1} and the full-rank of $\Sigma_{\theta^{\mathcal{U}_c}}$. Using this, the next result ensures consistent estimation of neighbors of $\mathcal{U}$.
\begin{theorem} \label{thm:U_identify_2_noise}
For node $i \in \mathcal{U}$, let $b^{i*}_{n}$ be the optimal solution in Problem~\eqref{P:regression_nonoise_2} under noise. If $\snr \geq16\sqrt{2}\s^2_{id}$, neighbors of $i$ are given by $\{j :j\in \mathcal{U}_c, {b^{i*}_n}_{j-|\mathcal{U}|} \geq .5/\s_{id}\}$, where $\s_{id}$, $\snr$ are defined in Eq.~\eqref{beta_min}.
\end{theorem}
\begin{proof}
Consider $j\in\mathcal{U}_c$. By Theorem \ref{thm:U_identify_2}, noiseless solution $b^{i*}$ has $b^{i*}_{j-|\mathcal{U}|} = \beta_{ij}/H_\beta(i,i)\geq 1/\s_{id}$ if $(ij)\in\mathcal{E}$, and $0$ otherwise. If $\snr \geq16\sqrt{2}\s^2_{id}$, using Eq.~\eqref{delta_inf}, $\|b^{i*}_n-b^{i*}\|_{\infty}< .5/\s_{id}$. Thus, entries in $b^{i*}_n$ for neighbors of $i$ are greater than $.5/\s_{id}$, and less than $.5/\s_{id}$ otherwise.
\end{proof}

It is worth mentioning that, estimation of nodes in $\mathcal{U}$ and identification of their neighbors are done separately in Theorems \ref{thm:U_identify_1} and \ref{thm:U_identify_2}. To complete this section, we discuss why this is necessary and whether both these steps can be combined.

\subsection{Joint identification and neighborhood estimation for $\mathcal{U}$}
Theorem \ref{thm:U_identify_1} identifies under-excited nodes in $\mathcal{U}$ by zero-cost solution for Problem~\eqref{P:regression_nonoise_1}. However even when Assumption~\ref{a:terminal_node} holds, the solution to~\eqref{P:regression_nonoise_1} need not be unique. For example, consider node $2\in\mathcal{U}$ in the grid in Fig.~\ref{fig:wrong_ab}. Combining $H^T_\beta(2,:),~H^T_\beta(4,:)$, a family of zero-cost solutions $y =\frac{H^T_\beta(2,:)}{H_\beta(2,2)}-\epsilon\frac{H^T_\beta(4,:)}{H_\beta(4,4)}$ is constructed for $i = 2$, where $0\leq\epsilon< \frac{H_\beta(4,4)}{H_\beta(2,2)}\min\left(\frac{\beta_{23}}{\beta_{43}},\frac{\beta_{25}}{\beta_{45}}\right)$ ensures $y_{-2}\leq 0$. Hence, the regression coefficients of Theorem~\ref{thm:U_identify_1} may reflect incorrectly on the neighbor set. Observe that the correct neighbors of $2$ are identified only if $\epsilon = 0$. Indeed, $\epsilon = 0$ is indirectly enforced in Problem~\eqref{P:regression_nonoise_2} (unlike Problem~\eqref{P:regression_nonoise_1}) by restricting the regression coefficients to nodes in $\mathcal{U}_c$. 

If additional assumptions are allowed, joint identification of zero-injection nodes and their neighborhood estimation in Problem~\eqref{P:regression_nonoise_1} is possible. The next result demonstrates this under a sufficient topological assumption, that holds trivially for radial grids. 

\begin{theorem}\label{thm:U_identify_3}
In grid $\mathcal{G}$ with Assumption~\ref{a:terminal_node}, $\forall i \in\mathcal{U}$, assume that no $j \in \mathcal{U}-\{i\}$ exists such that all neighbors of $j$ are also neighbors of $i$. Then, true neighbors of $i\in\mathcal{U}$ are given by the non-zero entries in the solution of Problem~\eqref{P:regression_nonoise_1}.
\end{theorem}
The proof is given in Appendix \ref{sec:lem:U_identify_3}. Note that nodes $2,4$ in Fig.~\ref{fig:wrong_ab} violate the condition in Theorem~\ref{thm:U_identify_3}, hence estimation of $2$'s neighbors needs Problem~\eqref{P:regression_nonoise_2}. The consistency of Theorem~\ref{thm:U_identify_3} under noisy measurements can be shown using a similar analysis as Theorem~\ref{thm:U_identify_2_noise}. In the next section, we complete the learning process by estimating the remaining grid edges between non-zero injection nodes.

\section{Learning edges between non-zero injection nodes} \label{sec:non-zero}
In this section, we construct a method to reconstruct the remaining graph, i.e., identifying the edges between non-zero injection nodes.
\subsection{Structural Assumptions}
First, we partition $H_{\beta},J_{\beta}$ into blocks for zero and non-zero nodes, $H_\beta = \begin{bmatrix}H^{\mathcal{U}\mathcal{U}}_{\beta} &H^{\mathcal{U}\mathcal{U}_c}_{\beta}\\H^{\mathcal{U}_c\mathcal{U}}_{\beta} &H^{\mathcal{U}_c\mathcal{U}_c}_{\beta}\end{bmatrix}$, $J_\beta = \begin{bmatrix}J^{\mathcal{U}\mathcal{U}}_{\beta} &J^{\mathcal{U}\mathcal{U}_c}_{\beta}\\J^{\mathcal{U}_c\mathcal{U}}_{\beta} &J^{\mathcal{U}_c\mathcal{U}_c}_{\beta}\end{bmatrix}$. Then,
\begin{align}
&\Sigma_{\theta^{\mathcal{U}_c}} = {J^{\mathcal{U}_c\mathcal{U}_c}_{\beta}}\Sigma_{p^{\mathcal{U}_c}}{J^{\mathcal{U}_c\mathcal{U}_c}_{\beta}} \text{~(by Eq.~\eqref{covar_DC}, and $p^{\mathcal{U}_c} = \mathbf{0}$)}\label{kron_-1}\\
&\text{where~~} {J^{\mathcal{U}_c\mathcal{U}_c}_{\beta}}^{-1} = H^{\mathcal{U}_c\mathcal{U}_c}_{\beta}-H^{\mathcal{U}_c\mathcal{U}}_{\beta}{H_{\beta}^{\mathcal{U}\mathcal{U}}}^{-1}H^{\mathcal{U}\mathcal{U}_c}_{\beta}.\label{kron_H}
\end{align}
Note that $H_{\beta}^{\mathcal{U}\mathcal{U}}$ and its inverse are diagonal matrices as nodes in $\mathcal{U}$ are not adjacent. Here, ${J^{\mathcal{U}_c\mathcal{U}_c}_{\beta}}^{-1}$ represents the weighted Laplacian for a Kron-reduced $\hat{\mathcal{G}}$ obtained by removing nodes in $\mathcal{U}$ from $\mathcal{G}$ \cite{kron}. See Fig.~\ref{fig:grid}(b) for an illustrative example. To identify edges between nodes in $\mathcal{U}_c$, we impose the following topological restriction for zero-injection nodes.
\begin{assumption} \label{a:girth}
The minimum size of a loop in $\mathcal{G}$ is $4$. Further, any loop of size $4$ doesn't include nodes in $\mathcal{U}$, while any loop of size $5$ includes at most one node in $\mathcal{U}$.
\end{assumption}

Note that Assumption~\ref{a:girth} holds trivially for radial grids, that include a majority of distribution grids. The minimum loop size of $4$ is non-restrictive in real grids with large girth \cite{NY,germany}. In fact it is necessary for consistent estimation even in the fully-excited setting \cite{dekaloopy} (see Lemma \ref{lemma:full_inv}). Further, as nodes in $\mathcal{U}$ have degree at least two, neighbors of any $j\in\mathcal{U}$ cannot all be neighbors of another node $i\in\mathcal{U}$, as that would create a loop of size $4$. Assumption~\ref{a:girth}, thus, implies Theorem~\ref{thm:U_identify_3} and enables joint identification of nodes in $\mathcal{U}$ and estimation of their neighbors. 

\subsection{Noiseless setting}
Consider nodes $i,j \in \mathcal{U}_c$, If $\exists k \in \mathcal{U}$ with edges $(ik),(jk)$ (estimated by Theorem~\ref{thm:U_identify_2}), then $(ij) \notin \mathcal{E}$ as that will create a loop of size $3$. To complete the topology learning, we, thus, need to estimate only edges between $i,j\in \mathcal{U}_c$ that do not have a common neighbor in $\mathcal{U}$. The following result identifies such edges, in the noiseless setting. 
\begin{figure}[htb]
\centering
\includegraphics[width=0.48\textwidth]{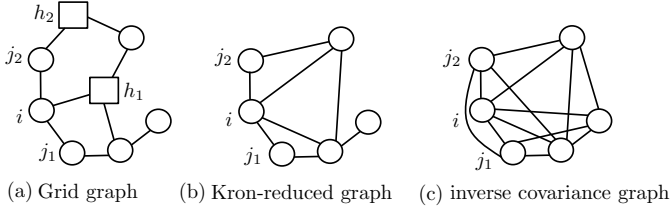}
\caption{(a) ${\mathcal{G}}$ with nodes in $\mathcal{U}$ marked square, $h_1, h_2$ violate Assumption~\ref{a:girth}, (b) $\hat{\mathcal{G}}$ after Kron-reducing $\mathcal{U}$, (c) Non-zero entries in inverse voltage covariance in $\mathcal{U}_c$.\label{fig:grid}}
\end{figure}
\begin{theorem}\label{thm:iden_U_c}
Consider nodes $i,j\in \mathcal{U}_c$ without any common neighbor in $\mathcal{U}$. Under Assumption $1,3$, $(ij) \in \mathcal{E}$ if and only if $\Sigma^{-1}_{\theta^{\mathcal{U}_c}}(i-|\mathcal{U}|,j-|\mathcal{U}|)<0$. 
\end{theorem}

\noindent Before proving the theorem, we state the following lemma.
\begin{lemma}  \label{lem:Jinv}
Let $i,j\in \mathcal{U}_c$ without a common neighbor in $\mathcal{U}$. Then ${J^{\mathcal{U}_c\mathcal{U}_c}_{\beta}}^{-1}(i-|\mathcal{U}|,j-|\mathcal{U}|) = H_\beta(i,j)$.
\end{lemma}

\begin{proof}
Consider $i,j\in \mathcal{U}_c$ without common neighbor in $\mathcal{U}$. The $(i-|\mathcal{U}|,j-|\mathcal{U}|)^\text{th}$ entry in ${J^{\mathcal{U}_c\mathcal{U}_c}_{\beta}}^{-1}$ corresponds to $(i,j)$ in Kron-reduced $\hat{\mathcal{G}}$. As $H_{\beta}^{\mathcal{U}\mathcal{U}}$ is diagonal under Assumption~\ref{a:terminal_node} and that $i$ and $j$ do not have a common neighbor in $\mathcal{U}$, it follows from Eq.~\eqref{kron_H} that ${J^{\mathcal{U}_c\mathcal{U}_c}_{\beta}}^{-1}(i-|\mathcal{U}|,j-|\mathcal{U}|) =H_\beta^{\mathcal{U}_c\mathcal{U}_c}(i-|\mathcal{U}|,j-|\mathcal{U}|)= H_\beta(i,j)$. 
\end{proof}

\createspace
\begin{proof}[Proof of Theorem~\ref{thm:iden_U_c}]
Using Lemma~\ref{lem:Jinv}, $(ij)$ is an edge in $\hat{\mathcal{G}}$, if and only if $(i,j)\in\mathcal{E}$. We now prove that $i,j$ are not part of a loop of size $3$ in $\hat{\mathcal{G}}$. The result then holds by applying Lemma~\ref{lemma:full_inv} in $\hat{\mathcal{G}}$.

If $(i,j)\notin\mathcal{E}$, then edge $(ij)$ doesn't exist in $\hat{\mathcal{G}}$ and $i,j$ are not part of a three node loop. Next let $(i,j) \in \mathcal{E}$. As minimum loop size in $\mathcal{G}$ is four, a three node loop with $i,j$ and some $k$ exists in $\hat{\mathcal{G}}$ if (A) $\{(ik),(jh),(kh)\}\subseteq\mathcal{E}$ or $\{(ih),(kh),(jk)\}\subseteq\mathcal{E}$, $\exists h \in \mathcal{U}$, or (B) $(ih_i),(kh_i),(jh_j),(kh_j)\subseteq\mathcal{E}$, $\exists h_i,h_j \in \mathcal{U}$ (see Fig.~\ref{fig:grid}(a)). (A) implies $h$ is in a four node loop, while (B) implies $h_i,h_j$ are in a five node loop. (A), (B) thus contradict Assumption~\ref{a:girth}. Hence $i,j$ are not part of a three node loop in $\hat{\mathcal{G}}$. Using Lemma \ref{lemma:full_inv} on $\hat{\mathcal{G}}$, $\Sigma^{-1}_{\theta^{\mathcal{U}_c}}(i-|\mathcal{U}|,j-|\mathcal{U}|)<0$ for edge $(ij)$, and $\geq 0$ otherwise.
\end{proof}
\textbf{Note:} (a) For the LC-PF Eq.~\eqref{LC_PF}, as per the remark following Lemma \ref{lemma:full_inv}, an edge between $(i+|\mathcal{U}|),(j+|\mathcal{U}|) \in \mathcal{U}_c$ without common neighbor in $\mathcal{U}$ is present \textit{iff} $\Sigma^{-1}_{(v^{\mathcal{U}_c},\theta^{\mathcal{U}_c})}(i,j)+\Sigma^{-1}_{(v^{\mathcal{U}_c},\theta^{\mathcal{U}_c})}(i+|\mathcal{U}_c|,j+|\mathcal{U}_c|) <0$.

(b) In loopy $\mathcal{G}$, violation of Assumption~\ref{a:girth} can create three node loops in the Kron-reduced $\hat{\mathcal{G}}$ (see Fig.~\ref{fig:grid}(b)). Lemma~\ref{lemma:full_inv} may fail to distinguish true edges in that setting, as shown through counter-examples in fully-excited grids with $3$-node loops in \cite{dekairep}. We now discuss Theorem \ref{thm:iden_U_c} in the presence of noisy measurements.

\subsection{Noisy setting} \label{sec:noise_2}
Under the noise model in Eq.~\eqref{eq_noise}, consider $\Tilde{\theta}^{\mathcal{U}_c} ={\theta}^{\mathcal{U}_c} + n^{\mathcal{U}_c}$ with noisy inverse covariance matrix  $\Sigma^{-1}_{\Tilde{\theta}^{\mathcal{U}_c}}$. Using Woodbury formula \cite{hager1989updating,dekaloopy}, we have 
\begin{align}
 \Delta\Sigma^{-1}_{\theta^{\mathcal{U}_c}} = \Sigma^{-1}_{\theta^{\mathcal{U}_c}}-\Sigma^{-1}_{\Tilde{\theta}^{\mathcal{U}_c}} = \Sigma^{-1}_{\theta^{\mathcal{U}_c}}(\Sigma^{-1}_{n^{\mathcal{U}_c}} + \Sigma^{-1}_{\theta^{\mathcal{U}_c}})^{-1}\Sigma^{-1}_{\theta^{\mathcal{U}_c}}, \label{deviateinverse}
\end{align}
 Using positive-definiteness of $\Delta\Sigma^{-1}_{\theta^{\mathcal{U}_c}}$ in Eq.~\eqref{deviateinverse}, we have
\begin{align}
&\max_{i,j} |\Delta\Sigma^{-1}_{\theta^{\mathcal{U}_c}}(i,j)|\leq \max_{i} \Delta\Sigma^{-1}_{\theta^{\mathcal{U}_c}}(i,i) \leq \sigma_{max}(\Delta\Sigma^{-1}_{\theta^{\mathcal{U}_c}})\nonumber\\
&\leq \frac{\sigma_{max}^2(\Sigma^{-1}_{\theta^{\mathcal{U}_c}})}{\sigma_{min}(\Sigma^{-1}_{n^{\mathcal{U}_c}} + \Sigma^{-1}_{\theta^{\mathcal{U}_c}})}<\sigma_{max}^2(\Sigma^{-1}_{\theta^{\mathcal{U}_c}})\sigma_{max}(\Sigma_{n^{\mathcal{U}_c}})&\nonumber\\
&\hspace{100pt}\leq \sigma_{max}(\Sigma_{n})/\sigma_{min}^2(\Sigma_{\theta^{\mathcal{U}_c}}).\label{maxdeviate}
\end{align}
We use this to modify Theorem~\ref{thm:iden_U_c} under noise.

\begin{theorem}\label{thm:iden_U_c_noise}
Consider $i,j \in \mathcal{U}_c$ with no common neighbor in $\mathcal{U}$. If $\snr \geq \frac{\max_i\Sigma_p(i,i)}{\beta^2_{\min}\sigma_{min}(\Sigma_{\theta^{\mathcal{U}_c}})}$, edge $(ij)$ exists if and only if $\Sigma^{-1}_{\tilde{\theta}^{\mathcal{U}_c}}(i-|\mathcal{U}|,j-|\mathcal{U}|)\leq-\beta^2_{\min}/\max_i\Sigma_p(i,i)$, where $\beta_{min},\snr$ are given in~\eqref{beta_min}.
\end{theorem}
\begin{proof}
Consider $i,j \in \mathcal{U}_c$ without a common neighbor in $\mathcal{U}$. For $(ij)\notin \mathcal{E}$, by Theorem~\ref{thm:iden_U_c}, $\Sigma^{-1}_{\theta^{\mathcal{U}_c}}(i-|\mathcal{U}|,j-|\mathcal{U}|)\geq 0$. If $\snr$ satisfies the stated bound, using Eq.~\eqref{maxdeviate}, $\Sigma^{-1}_{\tilde{\theta}^{\mathcal{U}_c}}(i-|\mathcal{U}|,j-|\mathcal{U}|) > -\frac{1}{\snr\sigma_{min}(\Sigma_{\theta^{\mathcal{U}_c}})}\geq -\beta^2_{\min}/\max_i\Sigma_p(i,i)$. 

For $(ij)\in \mathcal{E}$, using Lemma~\ref{lem:Jinv} and the positive-definiteness of ${J^{\mathcal{U}_c\mathcal{U}_c}_{\beta}}^{-1}$, we get ${J^{\mathcal{U}_c\mathcal{U}_c}_{\beta}}^{-1}(i-|\mathcal{U}|,i-|\mathcal{U}|) \geq \beta_{ij}$ and 
${J^{\mathcal{U}_c\mathcal{U}_c}_{\beta}}^{-1}(j-|\mathcal{U}|,j-|\mathcal{U}|) \geq \beta_{ij}$.
Using Eq.~\eqref{kron_H}, we get 
\begin{align*}
  \Sigma^{-1}_{\theta^{\mathcal{U}_c}}(i-|\mathcal{U}|,j-|\mathcal{U}|) &\leq -\beta_{ij}^2 (\Sigma^{-1}_p(i,i) + \Sigma^{-1}_p(j,j)) \\
  &\leq -2 \beta^2_{\min}/\max_i\Sigma_p(i,i).
\end{align*}
Using Eq.~\eqref{maxdeviate} and the stated SNR bound, $\Sigma^{-1}_{\tilde{\theta}^{\mathcal{U}_c}}(i-|\mathcal{U}|,j-|\mathcal{U}|) \leq -\beta^2_{\min}/\max_i\Sigma_p(i,i)$. 
\end{proof}
\createspace\textbf{Note:} Using Eq.~\eqref{kron_-1}, $\sigma_{max}(\Sigma^{-1}_{\theta^{\mathcal{U}_c}}) \leq \frac{\sigma_{max}(\Sigma^{-1}_{p^{\mathcal{U}_c}})}{\sigma^2_{min}(J^{\mathcal{U}_c\mathcal{U}_c}_{\beta})} \leq \frac{\sigma_{max}(\Sigma^{-1}_{p^{\mathcal{U}_c}})}{\sigma^2_{min}(J_\beta)} = \frac{\sigma^2_{max}(H_\beta)}{\min_{i\in\mathcal{U}_c}\Sigma_{p}(i,i)}$. Using this, Theorem \ref{thm:iden_U_c_noise} holds when $\snr \geq \frac{\max_i\Sigma_p(i,i)}{\min_{i\in\mathcal{U}_c}\Sigma_p(i,i)}(\frac{\sigma_{max}(H_\beta)}{\beta_{\min}})^2$
\\

\textbf{Learning algorithm:} The overall steps of our voltage based learning algorithm are listed in Algorithm $1$. Theorem~\ref{thm:U_identify_1} is used in Steps~\ref{step1_start}-\ref{step1_end} to identify zero-injection nodes. Then their neighborhood nodes are identified in Steps~\ref{step2_start}-\ref{step2_end} using Theorem~\ref{thm:U_identify_2}. Finally, edges between non-zero injection nodes are determined in Steps~\ref{step3_start}-\ref{step3_end} through Theorem~\ref{thm:iden_U_c}. An example of the learning steps for a test grid is given in Fig.~\ref{fig:algo1_eg}. 
\begin{algorithm}\label{alg:1}
\caption{Topology Learning for under-excited grids}
\textbf{Input:} $\theta$ samples for nodes in $\mathcal{V}$, positive thresholds $\tau_1,\tau_2,\tau_3$\\
\textbf{Output:} Edge set $\mathcal{E}$ in grid $\mathcal{G}$
\begin{algorithmic}[1]
\State $\mathcal{U} \gets \{\}$. $\mathcal{N}_2 \gets \{\}$
\ForAll{node $i$ in $\mathcal{V}$} \label{step1_start}
\State Solve Problem~\eqref{P:regression_nonoise_1} to get cost $c_i$.
\If{$c_i\leq \tau_1$} \label{step_tau_1}
\State $\mathcal{U} \gets \mathcal{U} \cup \{i\}$
\EndIf
\EndFor \label{step1_end}
\ForAll{node $i$ in $\mathcal{U}$} \label{step2_start}
\State Solve Problem~\eqref{P:regression_nonoise_2} to get solution $b^{i*}$.
\State $\mathcal{E} \gets \mathcal{E} \cup \{(ij): b^{i*}_{j-|\mathcal{U}|}\geq \tau_2\}$. \label{step_tau_2}
\State $\mathcal{N}_2 \gets \mathcal{N}_2 \cup \{(jk): b^{i*}_{j-|\mathcal{U}|},b^{i*}_{k-|\mathcal{U}|}\geq \tau_2\}$.
\EndFor \label{step2_end}
\State Compute $\Sigma^{-1}_{\theta^{\mathcal{U}_c}}$ using $\theta$ samples for nodes in $\mathcal{U}_c= \mathcal{V}-\mathcal{U}$.
\ForAll{$i\neq j \in \mathcal{U}_c$} \label{step3_start}
\If{$\Sigma^{-1}_{\theta^{\mathcal{U}_c}}(i,j)<-\tau_3~~\&~~ (ij) \notin \mathcal{N}_2$} \label{step_tau_3}
\State $\mathcal{E} \gets \mathcal{E} \cup \{(ij)\}$
\EndIf
\EndFor \label{step3_end}
\end{algorithmic}
\end{algorithm}

\textbf{Computational Complexity:} Solving each of the constrained regression problems ~\eqref{P:regression_nonoise_1}, and ~\eqref{P:regression_nonoise_2} for all nodes, needs $O(N^4)$ operations. Inverse covariance estimation takes $O(N^3)$ operations. Edge estimation and creation of $\mathcal{N}_2$ in Steps~\ref{step_tau_2}-\ref{step2_end} for all nodes in $\mathcal{U}$ is bounded by $O(N^3)$ operations. Adding edges between nodes in $\mathcal{U}_c$ needs $O(N^3)$ operations. The overall complexity is thus $O(N^4)$. 
\begin{figure}[htb]
\centering
\includegraphics[width=0.46\textwidth]{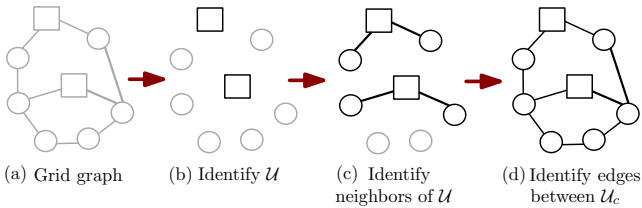}
\caption{Steps in Algorithm $1$ for a test grid.\label{fig:algo1_eg}}
\end{figure}

\textbf{Thresholds:} We include three thresholds: $\tau_1$ in Step~\ref{step_tau_1} (identify $\mathcal{U}$ nodes), $\tau_2$ in Step~\ref{step_tau_2} (identify neighbors of $\mathcal{U}$ nodes), and $\tau_3$ in Step~\ref{step_tau_3} (identify edges between $\mathcal{U}_c$ nodes). This is done to mitigate the effect of finite voltage samples and measurement noise. In the large sample limit, we combine results from Theorems~\ref{thm:U_identify_1_noise}, \ref{thm:U_identify_2_noise}, and \ref{thm:iden_U_c_noise} to determine the following thresholds, for consistent estimation in the presence of noise. 
\begin{theorem} \label{thm:overall_algorithm}
If $\snr \geq \max(\frac{\max_i\Sigma_p(i,i)}{\beta^2_{\min}\sigma_{min}(\Sigma_{\theta^{\mathcal{U}_c}})},16\sqrt{2}\s^2_{id}, 2(1+\s^4_{id}+\s^2_{id}))$, then taking $\tau_1 =\frac{\sigma_{min}(\Sigma_{\theta})}{1+\s_{id}^4+\s^2_{id}},\tau_2 = \frac{1}{2\s_{i}},\tau_3 = \frac{\beta^2_{\min}}{\max_i\Sigma_p(i,i)}$ in Algorithm $1$ gives asymptotically correct topology recovery, in the presence of noise, where $\snr, \s_{id}$ are defined in~\eqref{beta_min}.
\end{theorem}
In the next section, we present simulation results on the performance of our algorithm, in particular on noisy voltage data generated using real-world injection data through a non-linear power flow model.
\section{Numerical Simulations}
\label{sec:simulations}
We test Algorithm $1$ on two power grids, constructed from the IEEE $33$ bus system \cite{33bus}: (a) radial (Fig.~\ref{fig:radial}) with $9$ under-excited nodes, and (b) meshed/loopy (Fig.~\ref{fig:loopy}) with $8$ under-excited nodes, that respects Assumptions \ref{a:terminal_node} and \ref{a:girth}. For both grids, we test our algorithm using voltages generated using linear and non-linear models of (i) DC power flow and, (ii) AC power flow. While the linear models are given by Eqs.~\eqref{DC_PF},~\eqref{LC_PF}, the non-linear PF models are evaluated using Matpower \cite{zimmerman2011matpower}. For generating voltages in all cases, we consider nodal injection fluctuations that are uncorrelated across nodes. The fluctuations of nodal injections around their base loads are sampled using zero-mean Gaussian random variables of standard deviation $10^{-1}$. Eventually, we also test on voltage samples generated using real injection data from \cite{disc} that may be correlated between nodes. To demonstrate the performance under noise, we corrupt our generated voltage samples in all cases with zero-mean Gaussian noise of differing variance measured as a fraction of the variance of voltage measurements. It is worth noting that Algorithm $1$ only takes voltage measurements as input. No additional information regarding line parameters, values of injection statistics, or set of feasible lines are considered. In such a setting, the number of possible edges in our grid of $32$ non-reference nodes is $496$, with $32^{30}$ possible connected tree realizations (using Cayley's formula) and even more non-radial realizations. It is however noted that the performance of the algorithm can be improved by the inclusion of constraints if information regarding existence/non-existence of certain lines is known. 

We solve Problems~\eqref{P:regression_nonoise_1} and~\eqref{P:regression_nonoise_2} in Algorithm $1$ using CVX \cite{grant2014cvx}. We calibrate the algorithm's accuracy by relative estimation errors, computed as the sum of false edges and missed edges (false positives and true negatives) relative to the number of true edges in $\mathcal{E}$. The thresholds $\tau_1$, $\tau_2$ and $\tau_3$ in Algorithm $1$ are tuned to give minimal estimation error for a large sample size of $10^4$. These thresholds are then fixed, and the estimation errors at different sample sizes are computed by averaging over $15$ independent runs.

Fig.~\ref{fig:error_radialDC} shows the performance for the radial grid in Fig.~\ref{fig:radial} for linear and non-linear DC-PF voltages generated using simulated injection data. With noiseless data samples, the number of samples required for exact topology learning is approximately $300$ in either case. Perfect reconstruction under corrupted voltages with relative noise variance $1\%$ is observed at approximately $10^4$ samples. However, at $600$ samples, the average error is already below $5\%$ for both linear and non-linear models. While the performance improves with high sample sizes, note that the thresholds used for topology learning are not updated at each sample size, but fixed to the ones decided for asymptotically perfect recovery. In a realistic setting, the learning algorithm can consider the previous topology as a prior and select thresholds based on the pre-determined number of samples to improve the performance. Further the performance with linear and non-linear voltage samples are almost identical at sample sizes above $600$. This follows from the acceptable accuracy of DC-PF in modeling the non-linear counterpart. Fig.~\ref{fig:error_loopyDC} shows the performance of the algorithm for DC-PF for the loopy grid in Fig.~\ref{fig:loopy}. The performance, as expected, improves as the number of samples is increased. Beyond $300$ samples, the performance is comparable for linear and non-linear models under both noise and noiseless regimes.
\begin{figure}[tb]
\centering
\subfigure[]{\includegraphics[width=0.33\columnwidth]{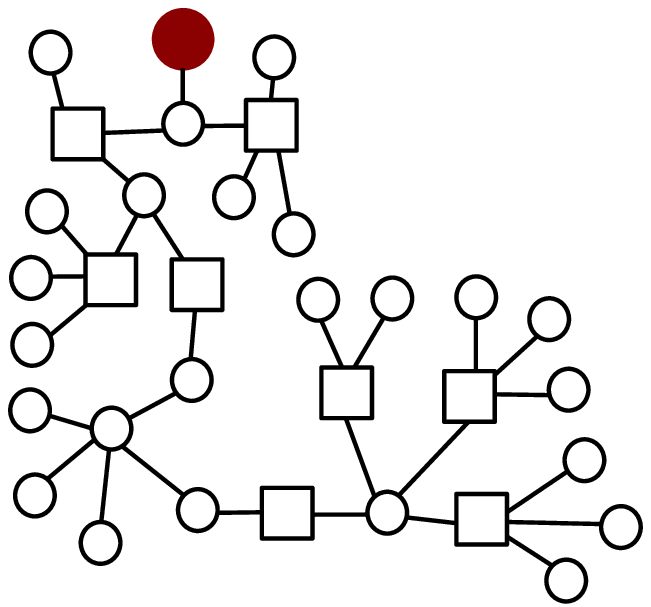}\label{fig:radial}}\hspace{30pt}
\subfigure[]{\includegraphics[width=0.33\columnwidth]{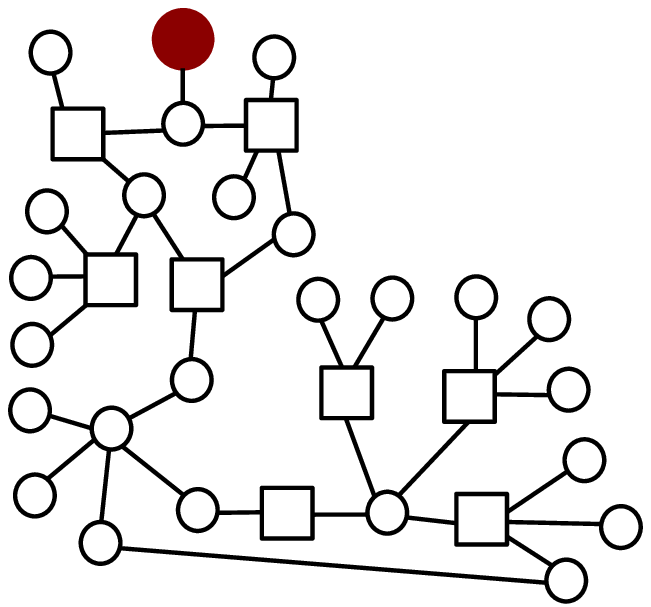}\label{fig:loopy}}\hfill
\caption{$33$ bus test networks \cite{33bus} for simulations, (a) radial (b) loopy. Non-zero injection nodes are marked square.}
\end{figure}
\begin{figure}[htb]
	\centering
	\subfigure[]{\includegraphics[width=0.76\columnwidth,height=0.56\columnwidth]{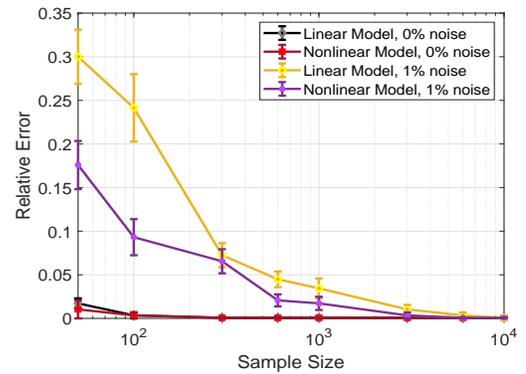} \label{fig:error_radialDC}}\hfill
	\subfigure[]{\includegraphics[width=0.76\columnwidth,height=0.56\columnwidth]{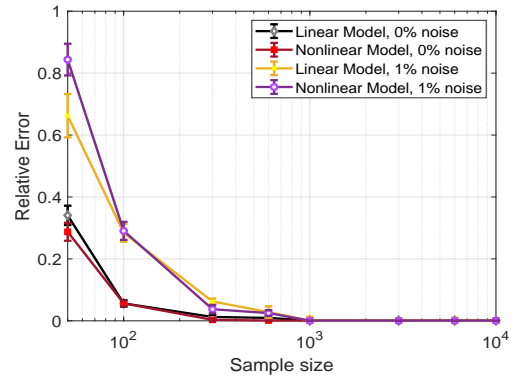}\label{fig:error_loopyDC}}
	\caption{Performance with linear and non-linear DC-PF voltages  for (a) radial grid in Fig.~\ref{fig:radial}, (b) loopy grid in Fig.~\ref{fig:loopy}.}
\end{figure}

Next, we demonstrate the performance of Algorithm $1$ for AC-PF voltage samples $(v,\theta)$, generated using the linear  
LC-PF Eq.~\eqref{LC_PF}, and the non-linear Matpower solver. The performance for the radial and loopy networks are given in Figs.~\ref{fig:error_radialAC} and~\ref{fig:error_loopyAC} respectively. In the noiseless setting, average errors go down to zero for the radial case at $600$ samples, and for the loopy case at $1000$ samples under both linear and non-linear models. With noise, the average errors are below $5\%$ for the radial case at $1000$ samples, and below $5\%$ for the loopy case at $1900$ samples. Perfect recovery in the noisy setting requires approximately $6000$ samples in the radial case and $3000$ samples in the loopy grid. 
\begin{figure}[htb]
	\centering
	\subfigure[]{\includegraphics[width=0.76\columnwidth,height=0.56\columnwidth]{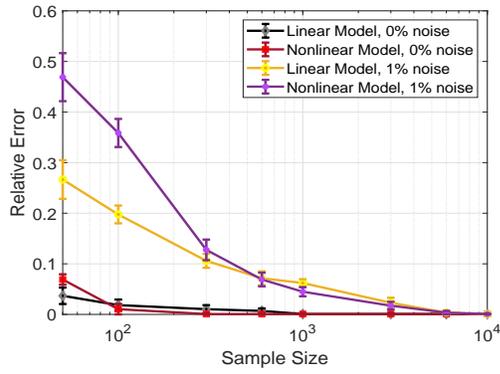} \label{fig:error_radialAC}}\hfill
	\subfigure[]{\includegraphics[width=0.76\columnwidth,height=0.56\columnwidth]{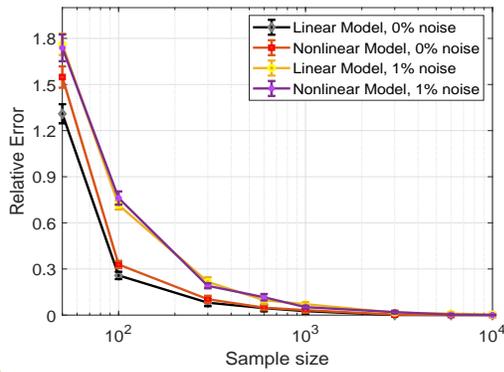}\label{fig:error_loopyAC}}
	\caption{Performance with linear and non-linear AC-PF voltages  for (a) radial grid in Fig.~\ref{fig:radial}, (b) loopy grid in Fig.~\ref{fig:loopy}.}
\end{figure}
Finally, we present results using real power/active load data, sampled at $15$ minute intervals from real house-holds \cite{disc}. Fig.~\ref{fig:error_real_loopyDC} demonstrates the performance of our learning algorithm for the loopy grid in Fig.~\ref{fig:loopy}, with linear and non-linear DC-PF samples. Note that the errors, in the noiseless setting, are below $2.67\%$ at $600$ samples for both the linear and non-linear case, despite no prior information about the topology. On the other hand, in the noisy setting, the average errors are below $2.67\%$ at $3000$ samples. All the experiments have asymptotically correct recovery ($0$ errors), which validates the theoretical consistency of our learning algorithm proven in the previous section. \begin{figure}[htb]
	\centering
	\begin{tabular}{c}
	\includegraphics[width=0.76\columnwidth,height=0.56\columnwidth]{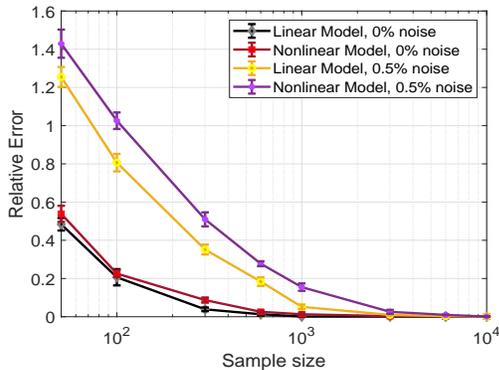}
	\end{tabular}
	\caption{Performance on loopy grid in Fig.~\ref{fig:loopy}, with linear and non-linear DC-PF voltages generated from real injection data.}
  \label{fig:error_real_loopyDC}\vspace{-10pt}
\end{figure}
\section{Conclusion}
\label{sec:conclusion}
This paper discusses topology learning using nodal voltages in general power grids under a realistic regime: one where a subset of the internal nodes in the grid are unexcited and have zero injection. This regime makes the voltage covariance matrix degenerate, and standard inverse covariance or mutual information based learning methods do not apply. We provide a novel learning algorithm to overcome the presence of unknown zero-injection nodes, by first identifying them and their neighboring nodes, and then identifying the remaining edges in a Kron-reduced network without the zero-injection nodes. Using techniques from noisy regression and inverse covariance estimation, we prove the asymptotic correctness of our learning algorithm, both in the noiseless and noisy settings, when unexcited nodes are internal, non-adjacent and not present in loops of size $4$. Notably, aside from voltage measurements, our algorithm does not use any prior information of structure, line parameters, values of nodal injections or their statistics. The theoretical contributions are validated by simulation results on IEEE test networks, with both linear and non-linear DC and AC voltage samples, corrupted with noise. Further, voltages generated using real injection data are used to caliberate the performance of our algorithm. 

While our work focuses on power grids, it naturally extends to learning under-excited structured equation models that follow mass and flow conservation laws. Additionally, this article focuses on voltage samples collected from a static and balanced power flow model. Extensions to three-phase networks are possibly using voltage statistics proposed for the fully-excited setting \cite{dekathreephase}. Similarly, extending our work to learning power system and general dynamical systems \cite{sauravacm,talukdar2020physics} in the under-excited regime, is another direction of future work.
\bibliography{voltage}
\appendix
\subsection{Proof of Lemma \ref{lem:boundU_c}}\label{sec:lem:boundU_c}
We decompose any $a \in\mathbb{R}$ as $a = a^+-a^-$, where 
$a^+ = \max(0,a)\geq 0$, and $a^- =-\min(0,a)\geq 0$. We list two inequalities that we use in the proof. 
\vspace{-2pt}\begin{align}
\forall a,b \in\mathbb{R}, &\text{~if~} b\geq 0, \text{~then~} (a+b)^+\geq a^+. \label{ineq_1}\\
&\text{~if~} a>b, \text{~then~} a^+\geq b^+. \label{ineq_2}\vspace{-2pt}
\end{align}
Consider $y$ with $y_i =1$, $y_{-i} = -a^{i*}$, where $a^{i*}>0$ is the optimal noiseless solution in Problem~\eqref{P:regression_nonoise_1} for $i\in\mathcal{U}_c$. Let $\sigma_{min^+}(\Sigma_{\theta})$ denote the smallest non-zero singular value of $\Sigma_{\theta}$. We have, 
\begin{align}
{y_n}^T(\Sigma_{\theta} + \Sigma_n)y_n \geq y_n^T\Sigma_{\theta}y_n\stackrel{a}{\geq}&y^T\Sigma_{\theta}y\geq \sigma_{min^+}(\Sigma_{\theta})\|y^R\|^2_2 \nonumber\\
&\stackrel{b}{\geq} \sigma_{min}(\Sigma_{\theta^{\mathcal{U}_c}})\|y^R\|^2_2 \label{ineq_U_c} \end{align}
where $(a)$ follows from optimality of $y$ in the noiseless setting. $y^R$ denotes the projection of $y$ in the range of $\Sigma_{\theta}$, and $(b)$ follows from Cauchy's Interlace Theorem as $\Sigma_{\theta^{\mathcal{U}_c}}$ is a principal sub-matrix of $\Sigma_{\theta}$. By Eq.~\eqref{pf2}, ${y^R}^T = y^T -c^TH^\mathcal{U}_\beta$ for some $c \in \mathbb{R}^{|\mathcal{U}|}$. Then 
\begin{equation}
\begin{split}
&\exists j_1,j_2 \in\mathcal{U},\text{~s.t.~} c^-_{j_1} = \max_{j\in\mathcal{U}} c_j^-,~c^+_{j_2} = \max_{j\in\mathcal{U}} c_j^+, \text{~and} \\
&\exists k_1\in\mathcal{U}_c-\{i\},\text{~s.t.~} (j_1k_1)\in\mathcal{E} \text{~(using Assumption~\ref{a:terminal_node})}. \label{def+-}
\end{split}
\end{equation}
Using the structure of $y$ and $H^{\mathcal{U}}_\beta$, we have, for $i\in\mathcal{U}_c$, {\footnotesize
\begin{align}
&\|y_R\|^2_2= \|y^T_R\|^2_2 = A +B +C \text{~where~}\label{long_f}\\
&A \coloneqq (1- \sum\limits_{j\in\mathcal{U}}c_jH_\beta(i,j))^2 = (1+ \sum\limits_{j\in\mathcal{U}}c^-_jH_\beta(i,j)-\sum\limits_{j\in\mathcal{U}}c^+_jH_\beta(i,j))^2\nonumber\\
&\hspace{6pt}\stackrel{(a_1)}{\geq}
{(1+\sum\limits_{j\in\mathcal{U}}c^-_jH_\beta(i,j))^+}^2\stackrel{(a_2)}{\geq} {(1-c^-_{j_1}H_\beta(i,i))^+}^2\label{A}\\
&B \coloneqq \sum\limits_{j\in\mathcal{U}}(-a_j -c_jH_\beta(j,j))^2= \sum\limits_{j\in\mathcal{U}}(c^+_jH_\beta(j,j)+a_j-c^-_jH_\beta(j,j))^2\nonumber\\
&\hspace{6pt}\stackrel{(b_1)}{\geq} (c^+_{j_2}H_\beta(j_2,j_2))^2\label{B}\\
&C\coloneqq \smashoperator[lr]{\sum\limits_{k\in\mathcal{U}_c-\{i\}}}(-a_k-\sum\limits_{j \in\mathcal{U}}c_jH_\beta(k,j))^2\geq (a_{k_1}+\sum\limits_{j \in\mathcal{U}}c_jH_\beta(k_1,j))^2\nonumber\\
&\hspace{6pt}\stackrel{(c_1)}{\geq} {(\sum\limits_{j \in\mathcal{U}}c^+_jH_\beta(k_1,j)-\sum\limits_{j \in\mathcal{U}}c^-_jH_\beta(k_1,j))^+}^2\nonumber\\
&\hspace{6pt}\stackrel{(c_2)}{\geq}{(-c^+_{j_2}H_\beta(k_1,k_1)-\sum\limits_{j \in\mathcal{U}}c^-_jH_\beta(k_1,j))^+}^2\nonumber\\
&\hspace{7pt}\geq{(c^-_{j_1}\beta_{k_1j_1}-c^+_{j_2}H_\beta(k_1,k_1))^+}^2 \nonumber\\
&\hspace{6pt}\stackrel{(c_3)}{\geq}{(c^-_{j_1}H_\beta(i,i)/\s_{id}-c^+_{j_2}H_\beta(j_2,j_2)\s_{id})^+}^2\label{C}
\end{align}
}
Here $(a_1),(b_1),(c_1)$ follow from Eq.~\eqref{ineq_1}, and $(a_2),(c_2)$ follow from Eqs.~\eqref{def+-},~\eqref{ineq_2}. $(c_3)$ follows from the definition of $\s_{id}$ in Eq.~\eqref{beta_min} and Eq.~\eqref{ineq_2}. Using Eqs.~\eqref{A},~\eqref{B},~\eqref{C} in Eq.~\eqref{long_f}, we have 
\begin{align}
  \|y_R\|^2_2 \geq {(1-x_1)^+}^2 + x_2^2+ {(x_1/\s_{id}-x_2\s_{id})^+}^2
\end{align}
where $x_1 = c^-_{j_1}H_\beta(i,i)\geq 0$ and $x_2 = c^+_{j_2}H_\beta(j_2,j_2)\geq 0$. Using the minimum value of quadratic functions, consider the cases for $x_1,x_2$ and their effect on $\|y_R\|^2_2$.\\
\begin{align}
  \text{(a)~} &x_2 \geq 1/\s_{id}^2:\nonumber\\ &\|y_R\|^2_2\geq 1/s^4_{id} \label{case_a}\\
\text{(b)~} &x_2 < 1/\s_{id}^2,~x_1\geq 1: \nonumber\\
&\|y_R\|^2_2 \geq \min\limits_{x_2} x_2^2+ (1/\s_{id}-x_2\s_{id})^2 = \frac{1}{s^4_{id}+s^2_{id}}\label{case_b}\\
\text{(c)~} &x_2 < 1/\s_{id}^2,~x_2s^2_{id}<x_1<1: \nonumber\\
&\|y_R\|^2_2 \geq \min\limits_{x_1, x_2} (1-x_1)^2 + x_2^2+ (x_1/\s_{id}-x_2\s_{id})^2\nonumber\\
&\hspace{24pt}= 1/(1+\s_{id}^4 +\s_{id}^2)\label{case_c}\\
\text{(d)~} &x_2 < 1/\s_{id}^2,~x_1\leq x_2s^2_{id}:\nonumber\\
&\|y_R\|^2_2 \geq \min\limits_{x_2} (1-x_2s^2_{id})^2+x_2^2 = 1/(1+\s_{id}^4)\label{case_d}
\end{align}
Using Eqs.~\eqref{case_a},~\eqref{case_b},~\eqref{case_c},~\eqref{case_d}, $\|y_R\|^2_2 \geq 1/(1+\s_{id}^4 +\s_{id}^2)$. Using this in Eq.~\eqref{ineq_U_c} completes the proof. 

\subsection{Proof of Theorem~\ref{thm:U_identify_3}}\label{sec:lem:U_identify_3}
Consider $i\in \mathcal{U}$. From the proof of Theorem \ref{thm:U_identify_1}, note that zero-cost solutions in Problem~\eqref{P:regression_nonoise_1} are given by $y^T = c^TH^{\mathcal{U}}_\beta$, where $y_i=1, y_{-i}<0$. 

$y_{j}\leq 0 \forall j \in \mathcal{U}-\{i\}$, necessitates $c_j\leq 0,\forall j\neq i$. Consider $j_1\in\mathcal{U}-\{i\}$. Under the assumption in the statement, $\exists k_1\in\mathcal{U}_c$ such that $(k_1j_1) \in \mathcal{E}$, but $(k_1i) \notin \mathcal{E}$. Thus, $y_{k_1} = -c_{j_1}\beta_{j_1k_1} + \sum_{j\in \mathcal{U}-\{i,j_1\}}c_jH_\beta(j,k_1)$, where $H_\beta(j,k_1) \leq 0$ follows from Eq.~\eqref{H_dd}. As $c_{j}$ is non-positive for $j\neq i$, $y_{k_1} \leq 0$ holds only if $c_{j_1}= 0$. Following the same reasoning, $\forall j \in \mathcal{U}-\{i\}$, $c_{j} =0$. Finally, $c_i = 1/H_\beta(i,i)$ is needed for $y_i =1$. Thus, $y = \frac{H^T_\beta(i,:)}{H_\beta(i,i)}$ is the unique zero-cost solution for $i \in \mathcal{U}$. It identifies the true neighbors of $i$, due to the structure of $H^T_\beta(i,:)$. 

\end{document}